\documentclass[11pt, reqno]{amsart}%
\usepackage{graphics,color}
\usepackage{amsfonts}
\usepackage{amsmath}
\usepackage{amssymb}
\usepackage{graphicx}
\usepackage[bookmarks=false]{hyperref}%
\newtheorem{thm}{Theorem}[section]
\newtheorem{theorem}{Theorem}[section]
\newtheorem{lemma}[thm]{Lemma}
\newtheorem{cor}[thm]{Corollary}

\newtheorem{defn}[thm]{Definition}

\numberwithin{equation}{section}
\begin{document}
\title[Sensitivity and block sensitivity of nested canalyzing function]{Sensitivity and block sensitivity of nested canalyzing function}
\author[Yuan Li and John O. Adeyeye]{Yuan Li$^{1\ast}$ and John O. Adeyeye $^{2\ast}$}

\address{{\small $^{1}$Department of Mathematics, Winston-Salem State University, NC
27110,USA}\\
{\small email: liyu@wssu.edu }\\
$^{2}$Department of Mathematics, Winston-Salem State University, NC 27110,USA,
{\small email: adeyeyej@wssu.edu}}

\thanks{$^{\ast}$ Supported by an award from the USA DoD $\#$ W911NF-11-10166}
\date{}

\begin{abstract}
Based on a recent characterization of nested canalyzing function (NCF), we obtain the formula of the sensitivity of any NCF. Hence we find that any sensitivity of NCF is between $\frac{n+1}{2}$ and $n$. Both lower and upper bounds are tight. We prove that the block sensitivity, hence the $l$-block sensitivity, is same to the sensitivity. It is well known that monotone function also has this property. We eventually find all the functions which are both monotone and nested canalyzing (MNCF). The cardinality of all the MNCF is also provided.

\end{abstract}
\maketitle

\section{Introduction}

\label{sec-intro} 
Nested Canalyzing Functions (NCFs) were introduced recently in \cite{Kau2}.
One important characteristic of (nested) canalyzing functions is that they
exhibit a stabilizing effect on the dynamics of a system. That is, small
perturbations of an initial state should not grow in time and must eventually
end up in the same attractor of the initial state. The stability is typically
measured using so-called Derrida plots which monitor the Hamming distance
between a random initial state and its perturbed state as both evolve over
time. If the Hamming distance decreases over time, the system is considered
stable. The slope of the Derrida curve is used as a numerical measure of
stability. Roughly speaking, the phase space of a stable system has few
components and the limit cycle of each component is short.

It is shown in \cite{Jar} that the class of nested canalyzing functions is identical to the class of so-called unate
cascade Boolean functions, which has been studied extensively in engineering
and computer science. It was shown in \cite{But} that this class produces the
binary decision diagrams with the shortest average path length. Thus, a more
detailed mathematical study of this class of functions has applications to
problems in engineering as well.

In \cite{Coo}, Cook et al. introduced the notion of sensitivity as a combinatorial 
measure for Boolean functions providing lower bounds on the time needed by CREW PRAM (concurrent read
 , but exclusive write (CREW) parallel random access machine (PRAM)). It was extended by Nisan \cite{Nis} to block
 sensitivity. It is still open whether sensitivity and block sensitivity are polynomially related. Although the definition is straightforward, the sensitivity is understood only for a few classes function. For monotone function, The block sensitivity is same as its sensitivity.
 
 Recently, in \cite{Yua1}, a complete characterization for nested canalyzing
function is obtained by obtaining its unique algebraic normal form. A new concept, called  LAYER NUMBER, is introduced. Based on this,    explicit formulas for number of all the nested canalyzing functions and the average sensitivity of any NCF are provided.  Theoretically, It is showed why NCF is stable since the upper bound is a constant.

In this paper, we obtain the formula of the sensitivity of any NCF based on a characterization of NCF from \cite{Yua1}. We show the block sensitivity, like monotone function, is also same to its sensitivity. Finally, we characterize all the Boolean functions which are both nested canalyzing and monotone. We also give the number of such functions.
\section{Preliminaries}

\label{2} In this section we introduce the definitions and notations. Let
$\mathbb{F}=\mathbb{F}_{2}$ be the Galois field with $2$ elements. If $f$ is a
$n$ variable function from $\mathbb{F}^{n}$ to $\mathbb{F}$, it is well known
\cite{Lid} that $f$ can be expressed as a polynomial, called the algebraic
normal form(ANF):
\[
f(x_{1},x_{2},\ldots,x_{n})=\bigoplus_{0\leq k_i\leq 1,i=1,\ldots,n}a_{k_{1}k_{2}\ldots k_{n}}{x_{1}}^{k_{1}}{x_{2}}^{k_{2}%
}\cdots{x_{n}}^{k_{n}}%
\]
where each coefficient $a_{k_{1}k_{2}\ldots k_{n}}\in\mathbb{F}$ is a
constant. 

\begin{defn}\label{def2.1} Let $f$ be a Boolean function in $n$ variables. Let $\sigma$ be
a permutation on $\{1,2,\ldots,n\}$. The function $f$ is nested canalyzing
function (NCF) in the variable order

$x_{\sigma(1)},\ldots,x_{\sigma(n)}$ with canalyzing input values
$a_{1},\ldots,a_{n}$ and canalyzed values $b_{1},\ldots,b_{n}$, if it can be
represented in the form

$f(x_{1},\ldots,x_{n})=\left\{
\begin{array}
[c]{ll}%
b_{1} & x_{\sigma(1)}=a_{1},\\
b_{2} & x_{\sigma(1)}= \overline{ a_{1}}, x_{\sigma(2)}=a_{2},\\
b_{3} & x_{\sigma(1)}= \overline{ a_{1}}, x_{\sigma(2)}= \overline{ a_{2}},
x_{\sigma(3)}=a_{3},\\
.\ldots. & \\
b_{n} & x_{\sigma(1)}= \overline{ a_{1}}, x_{\sigma(2)}= \overline{ a_{2}%
},\ldots,x_{\sigma(n-1)}= \overline{ a_{n-1}}, x_{\sigma(n)}=a_{n},\\
\overline{b_{n}} & x_{\sigma(1)}= \overline{ a_{1}}, x_{\sigma(2)}= \overline{
a_{2}},\ldots,x_{\sigma(n-1)}= \overline{ a_{n-1}}, x_{\sigma(n)}=\overline{
a_{n}}.
\end{array}
\right. $

Where $\overline{a}=a\oplus 1$.The function f is nested canalyzing if f is nested
canalyzing in the variable order $x_{\sigma(1)},\ldots,x_{\sigma(n)}$ for some
permutation $\sigma$.
\end{defn}

\begin{theorem}\label{th1}\cite{Yua1}
 Given $n\geq2$, $f(x_{1},x_{2},\ldots,x_{n})$ is nested
canalyzing iff it can be uniquely written as
\begin{equation}\label{eq2.1}
f(x_{1},x_{2},\ldots,x_{n})=M_{1}(M_{2}(\ldots(M_{r-1}%
(M_{r}\oplus 1)\oplus 1)\ldots)\oplus 1)\oplus b.
\end{equation}
Where  $M_{i}=\prod_{j=1}^{k_{i}}(x_{i_{j}}\oplus a_{i_{j}})$,
$i=1,\ldots,r$, $k_{i}\geq1$ for $i=1,\ldots,r-1$, $k_{r}\geq2$, $k_{1}%
+ \ldots + k_{r}=n$, $a_{i_{j}}\in\mathbb{F}_{2}$, $\{i_{j}|j=1,\ldots,k_{i},
i=1,\ldots,r\}=\{1,\ldots,n\}$.
\end{theorem}

Because each NCF can be uniquely written as \ref{eq2.1} and the number $r$ is
uniquely determined by $f$, we have

\begin{defn}\label{def2.2} For a NCF $f$ written as equation \ref{eq2.1}, the number $r$ will
be called its LAYER NUMBER. Variables of $M_{1}$ will be called the
most dominant variables(canalyzing variable), they belong to the first layer of this NCF.
Variables of $M_{2}$ will be called the second most dominant variables and belong to the second layer 
of this NCF and etc. We Call $[k_1,\ldots,k_r]$ the profile of $f$. There are $2^{n+1}$ NCFs with the same profile.
\end{defn}

\section{Sensitivity and block sensitivity of NCF}
Let $\mathbf{x}=(x_1,\ldots,x_n)\in \mathbb{F}^n$, $[n]=\{1,\ldots,n\}$. For any subset $S$ of $[n]$, we form $\mathbf{x}^S$ by complementing those bits in $\mathbf{x}$ indexed by elements of $S$. We write $\mathbf{x}^i$ for $\mathbf{x}^{\{i\}}$.

\begin{defn}
The sensitivity of $f$ at $\mathbf{x}$, $s(f;\mathbf{x})$, is the number of indices $i$ such that $f(\mathbf{x})\neq f(\mathbf{x}^i)$.
 The sensitivity of $f$, denoted $s(f)$, is $Max_{\mathbf{x}}s(f;\mathbf{x})$
\end{defn}
\begin{defn}\cite{Nis}
The block sensitivity of $f$ at $\mathbf{x}$, $bs(f;\mathbf{x})$, is the maximum number of disjoint subsets $B_1,\ldots,B_r$ of $[n]$ such that, for all $j$, $f(\mathbf{x})\neq f(\mathbf{x}^{B_j})$. We refer to such a set $B_j$ as a block. The block sensitivity of $f$, denoted $bs(f)$, is $Max_{\mathbf{x}}bs(f;\mathbf{x})$.
\end{defn}

\begin{defn}\cite{Ken}
The $l$-block sensitivity of $f$ at $\mathbf{x}$, $bs_{l}(f;\mathbf{x})$, is the maximum number of disjoint subsets $B_1,\ldots,B_r$ of $[n]$ such that, for all $j$, $B_j\leq l$ and $f(\mathbf{x})\neq f(\mathbf{x}^{B_j})$.  The $l$-block sensitivity of $f$, denoted $bs_{l}(f)$, is $Max_{\mathbf{x}}bs_l(f;\mathbf{x})$.
\end{defn}

Obviously, we have $0\leq s(f)\leq bs_l(f)\leq bs(f)\leq n$.

\begin{lemma}\label{lm3.1}
Let $\sigma$ be a permutation on $[n]$, and $(a_1\ldots,a_n)\in \mathbb{F}^n$, Let $g=f(x_{\sigma(1)},\ldots,x_{\sigma(n)})$ and $h=f(x_1\oplus a_1,\ldots, x_n\oplus a_n)$. Then the sensitivity, $l$-block sensitivity and block sensitivity of $f$, $f\oplus 1$, $g$, $h$ are same.
\end{lemma}
\begin{proof}
This follows from the above definitions.
\end{proof}
Because of Lemma \ref{lm3.1}, In the rest of this section, for NCF in equation \ref{eq2.1}, we always assume
\begin{equation*}
f(x_{1},x_{2},\ldots,x_{n})=f_r=M_{1}(M_{2}(\ldots(M_{r-1}%
(M_{r}\oplus 1)\oplus 1)\ldots)\oplus 1)
\end{equation*}

and $M_1=x_1\ldots x_{k_1}$, $M_2=x_{k_1+1}\ldots x_{k_1+k_2}$,$\ldots$, $M_r=x_{k_1+\ldots+k_{r-1}+1}\ldots x_n$

Let $\mathbf{x}=(\mathbf{x_1},\ldots,\mathbf{x_r})$, where $\mathbf{x_i}$ has $k_i$ bits, $i=1,\ldots,r$.
We have 
\begin{lemma}\label{lm3.2}
For NCF $f$, for any word $\mathbf{x}$, if there are more than one zero bits in a subword $\mathbf{x_i}$, then we keep only one zero and flip the others to get a new word $\mathbf{x'}$, we have $s(f;\mathbf{x})\leq s(f;\mathbf{x'})$(In fact,   $s(f;\mathbf{x'})=s(f;\mathbf{x})$ or $s(f;\mathbf{x'})=s(f;\mathbf{x})+1$) and $bs(f;\mathbf{x})= bs(f;\mathbf{x'})$
\end{lemma}
\begin{proof}
There are at least two zeros in $\mathbf{x_i}$, so $M_i$ is always zero (hence, does not change) even if one of the bits is flipped. hence, $s(f;\mathbf{x})\leq s(f;\mathbf{x'})$. 

Let $bs(f;\mathbf{x})=t$, and $B_j$, $j=1,\ldots,t$ be the blocks such that $f(\mathbf{x}^{B_j})\neq f(\mathbf{x})$. We can assume that each block $B_j$ is minimal, i.e., for any proper subset $B_j'\subset B_j$, $f(\mathbf{x}^{B_j'})= f(\mathbf{x})$. Suppose there is a block $B_{j_0}$ involves the bits in $\mathbf{x_i}$, it means it changes the value of $M_i$ from 0 to 1. It must change all the zero bits in 
$\mathbf{x_i}$ to 1. Such $B_{j_0}$ is unique since all the blocks are disjoint. We can construct the block $B_{j_0}'$(a subset of $B_{j_0}$), which has only one index whose corresponding bit is the only zero bit of $\mathbf{x_i'}$. Take all the other blocks same as $B_j$ ($j\neq j_0$). We get the value of $bs(f,\mathbf{x'})\geq t=bs(f,\mathbf{x})$. On the other hand, there are more zeros in $\mathbf{x_i}$, in order to change the value of $M_i$ (hence, a possible change of $f$) from 0 to 1, it needs to change more than one bits,  hence the number of maximal blocks will be probably less (or same), i.e., we have $bs(f,\mathbf{x})\geq bs(f,\mathbf{x'})$. Hence, $bs(f,\mathbf{x})= bs(f,\mathbf{x'})$
\end{proof}
We are ready to prove the main result of this paper, we have
\begin{thm}\label{th3.1}
 $f_r$ is nested canalyzing with profile $[k_1,\ldots,k_r]$, then 

$s(f_1)=n$.
For $r>1$,
$s(f_r)=\left\{
\begin{array}
[c]{ll}%
Max\{k_1+k_3+\ldots+k_r,k_2+k_4+\ldots k_{r-1}+1\},2\nmid r\\
Max\{k_1+k_3+\ldots+k_{r-1}+1,k_2+k_4+\ldots k_{r}\},2|r\\
\end{array}
\right. $
\end{thm}

\begin{proof}
It is obvious that $S(f_1)=n$.

For $r>1$, we first consider that $r$ is odd.

Let $s(f_r,\mathbf{x})$ be the sensitivity of $f_r$ on word $\mathbf{x}$ for $\mathbf{x}=(\mathbf{x_1},\ldots,\mathbf{x_r})$. Because of Lemma \ref{lm3.2}, in order to find the maximal value, we can assume that there is either no zero or exactly one zero bit in every $\mathbf{x_i}$.

In the following, we consider all the possibilities of such words $(\mathbf{x_1},\ldots,\mathbf{x_r})$.

Case 1: One zero in $\mathbf{x_1}$.

$f=0$, in order to change the value, the zero bit in $\mathbf{x_1}$ must be changed. Hence, $S(f_r,\mathbf{x})\leq 1$.

Case 2: No zero in $\mathbf{x_1}$, but one zero in $\mathbf{x_2}$. 

$f_r=M_1(M_2(\ldots)\oplus 1)=M_1=1$,  the value of $f_r$ does not change by flipping any bit in $\mathbf{x_i}$( $i\geq 3$)or  any nonzero bits in $M_2$. Hence, $s(f_r,\mathbf{x})\leq k_1+1$.

Case 3: No zero in $\mathbf{x_1}$ and $\mathbf{x_2}$, but a zero in $\mathbf{x_3}$. 

$f_r=M_1(M_2(M_3(\ldots)\oplus 1)\oplus 1)=M_1(M_2\oplus 1)=0$. In order to change the value of $f_r$ (from 0 to 1), we can only flip the bits in $\mathbf{x_2}$ or possibly the zero bit in $\mathbf{x_3}$, hence, $s(f_r,\mathbf{x})\leq k_2+1$.

Case 4: No zero in $\mathbf{x_1}$, $\mathbf{x_2}$ and $\mathbf{x_3}$, but a zero in $\mathbf{x_4}$. 

$f_r=M_1(M_2(M_3\oplus 1)\oplus 1)=1$.
 We can change the value of $f_r$ (from 1 to 0)by flipping any bit in $\mathbf{x_1}$ or $\mathbf{x_3}$(or possible the zero bit in $\mathbf{x_4}$) but not the bit in $\mathbf{x_2}$ and all the $\mathbf{x_i}$ with $i\geq 5$.
  Hence, we have  $S(f_r,\mathbf{x})\leq k_1+k_3+1$.
  
  $\ldots$ $\ldots$ $\ldots$
  
Case $r$: No zero in $\mathbf{x_i}$, $i=1,\ldots,r-1$ but one zero in $\mathbf{x_{r}}$.

$f_r=M_{1}(M_{2}(\ldots(M_{r-1}\oplus 1)\ldots)\oplus 1)=0$ (since $r-1$ is even). We can change the value of $f_r$ (from 0 to 1) by flipping one bit of any $\mathbf{x_2}$, $\mathbf{x_4}$, $\ldots$, $\mathbf{x_{r-1}}$ and the zero bit of $\mathbf{x_{r}}$ but not the bit of $\mathbf{x_1}$, $\mathbf{x_3}$, $\ldots$, $\mathbf{x_{r}}$. Hence, we have  $S(f_r,\mathbf{x})= k_2+k_4+\ldots+k_{r-1}+1$

Case $r+1$: No zero in $\mathbf{x_i}$, $i=1,\ldots,r$, i.e., $\mathbf{x}=(1,\ldots,1)$.
$f_r=M_{1}(M_{2}(\ldots(M_{r}\oplus 1)\ldots)\oplus 1)=1$ (since $r$ is odd). We can change the value of $f_r$ (from 1 to 0) by flipping one bit of any $\mathbf{x_1}$, $\mathbf{x_3}$, $\ldots$, $\mathbf{x_{r}}$ but not the bit of $\mathbf{x_2}$, $\mathbf{x_4}$, $\ldots$, $\mathbf{x_{r-1}}$. Hence, we have  $S(f_r,\mathbf{x})= k_1+k_3+\ldots+k_{r}$

In summary, we have $s(f_r)=Max\{k_2+k_4+\ldots+k_{r-1}+1, k_1+k_3+\ldots+k_{r}\} $.

When $r$ is even, the proof is similar.

\end{proof}
\begin{cor}\label{Co3.1}
$s(f_1)=n$.
$s(f_{n-1})=\left\{
\begin{array}
[c]{ll}%
\frac{n+2}{2},2|n\\
\frac{n+1}{2},2\nmid n\\
\end{array}
\right. $.

If $2\leq r\leq n-2$,
$\frac{n+1}{2}\leq S(f_{r})\leq \left\{
\begin{array}
[c]{ll}%
n+1-\frac{r+1}{2},2\nmid r\\
n+1-\frac{r}{2},2| r\\
\end{array}
\right. $

\end{cor}

\begin{proof}
Because $(k_2+k_4+\ldots+k_{r-1}+1)+(k_1+k_3+\ldots+k_{r})=n+1$, hence,  $Max\{k_2+k_4+\ldots+k_{r-1}+1, k_1+k_3+\ldots+k_{r}\}|\geq \frac{n+1}{2} $. By considering the two minimal possibilities of $(k_2+k_4+\ldots+k_{r-1}+1)$ and $(k_1+k_3+\ldots+k_{r})$, we will get the maximal valus of these two numbers. Hence, we can get the other side of the above inequality.

\end{proof}

In the following , we will prove the block sensitivity of any NCF is same to its sensitivity.
Because of Lemma \ref{lm3.1}, we still assume
\begin{equation*}
f(x_{1},x_{2},\ldots,x_{n})=f(\mathbf{x})=f(\mathbf{x_1},\ldots,\mathbf{x_r})=f_r=M_{1}(M_{2}(\ldots(M_{r-1}%
(M_{r}\oplus 1)\oplus 1)\ldots)\oplus 1)
\end{equation*}

and $M_1=x_1\ldots x_{k_1}$, $M_2=x_{k_1+1}\ldots x_{k_1+k_2}$,$\ldots$, $M_r=x_{k_1+\ldots+k_{r-1}+1}\ldots x_n$

\begin{thm}\label{th3.2}
Let $f$ be any NCF, then $s(f)=bs(f)$.
\end{thm}
\begin{proof}
Actually, by Lemma \ref{lm3.2}, we just need to prove  $s(f_r;\mathbf{x})=bs(f_r;\mathbf{x})$ for any $\mathbf{x}$ such that there is at most one zero bit in each subword $\mathbf{x_i}$.
If $r=1$, since $s(f_1)=n\leq bs(f_1)\leq n$, we have $bs(f_1)=n$.
In the following we assume $r\geq 2$.
 For any word $\mathbf{x}$, let the first zero bit of $\mathbf{x}$ appear in $\mathbf{x_i}$, i.e., $M_1=\ldots=M_{i-1}=1$ and $M_i=0$.
So, we have 
\begin{equation}\label{eq3.1}
f_r=M_{1}(M_{2}(\ldots(M_{i-2}(M_{i-1}\oplus 1)\oplus 1)\ldots)\oplus 1)=\left\{
\begin{array}
[c]{ll}%
1,2|i\\
0,2\nmid i\\
\end{array}
\right. 
\end{equation}

Let $bs(f_r;\mathbf{x})=t$, and $B_j$, $j=1,\ldots,t$ be the disjoint blocks such that $f_r(\mathbf{x}^{B_j})\neq f_r(\mathbf{x})$. We can assume that each block $B_j$ is minimal, i.e., for any proper subset $B_j'\subset B_j$, $f_r(\mathbf{x}^{B_j'})= f_r(\mathbf{x})$.
First, all the blocks do not involve the bits of $M_j$ with $j\geq i+1$ because of \ref{eq3.1}. To change the value of $f_r$, some $M_l$ must be changed (from 1 to 0) for $l=1,\ldots,i-1$ or $M_i$ be changed from 0 to 1. In order to do so, we need only to flip one bit in $M_l$ ($l=1,\ldots,i-1$) from 1 to 0 or change the zero bit in $M_i$ to 1.  Hence the corresponding block $B_j$ has only one index since it is minimal. We actually have proved $s(f_r;\mathbf{x})\geq t$, hence $s(f_r;\mathbf{x})=bs(f_r;\mathbf{x})$.

\end{proof}

\section{Monotone nested canalyzing functions}
In this section, we determine all the functions which are both monotone and nested canalyzing.

\begin{defn}
Let $\mathbf{x}=(x_1,\ldots x_n)\in \mathbb{F}_2^n$ and $\mathbf{y}=(y_1,\ldots y_n)\in \mathbb{F}_2^n$, we define  $\mathbf{x}\prec \mathbf{y}$ iff $x_i\leq y_i$ for all $i\in [n]$.
\end{defn}

\begin{defn}
$f(\mathbf{x})$ is monotone increasing (decreasing) if $f(\mathbf{x})\leq f(\mathbf{y})$ ($f(\mathbf{x})\geq f(\mathbf{y})$) whenever $\mathbf{x}\prec \mathbf{y}$. 
\end{defn}

\begin{lemma}\label{lm4.1}
If $f$ is monotone increasing (decreasing), then fix the values of some bits, the remain function  of the remaining variable is still monotone increasing (decreasing).
\end{lemma}
\begin{lemma}\label{lm4.2}
$f(\mathbf{x})=(x_1\oplus a_1)\ldots(x_n\oplus a_n)\oplus b$ is monotone iff $a_1=\ldots=a_n$.
\end{lemma}

\begin{lemma}\label{lm4.3}
$f$ and $g$ are monotone increasing (decreasing) then $fg$ is also increasing (decreasing). $f\oplus 1$ will be decreasing (increasing).
\end{lemma}
Let $f_r$ be a NCF and written as \ref{eq2.1}.
\begin{thm}\label{th4.1}
$f_r$ is monotone iff $M_{i}=\prod_{j=1}^{k_{i}}(x_{i_{j}}\oplus a)$ and $M_{i+1}=\prod_{j=1}^{k_{i+1}}(x_{{i+1}_{j}}\oplus \overline{a})$ for $i=1,3,5,\ldots$. Where $a \in\{0,1\}$
\end{thm}
\begin{proof}
By suitably fixing the values of the other variables, we can get $f_r=M_i\oplus 1$ for $i>1$ or $M_1$. Hence,  by Lemma \ref{lm4.1} and Lemma \ref{lm4.2},  we have $M_{i}=\prod_{j=1}^{k_{i}}(x_{i_{j}}\oplus a)$. Again, we may suitably fix the values of the other variables  to  get $f_r=M_i(M_{i+1}\oplus 1)$. If $M_{i+1}=\prod_{j=1}^{k_{i}}(x_{i_{j}}\oplus a)$, $M_i(M_{i+1}\oplus 1)$ is not monotone. Hence,
$M_{i+1}=\prod_{j=1}^{k_{i+1}}(x_{{i+1}_{j}}\oplus \overline{a})$.

On the other hand, use induction principle , it is easy to prove these NCFs are monotone with the help of the above three lemmas.

Actually, When $M_1=x_1\ldots x_{k_1}$ , $f_r$ is increasing, when $M_1=(x_1\oplus 1)\ldots (x_{k_1}\oplus 1)$ , $f_r$ is decreasing.
\end{proof}

\begin{cor}\label{co4.1}
The number of monotone nested canalyzing functions (MNCFs) is 

\[
=4\sum_{\substack{k_{1}+\ldots+k_{r}=n\\k_{i}\geq1,i=1,\ldots,r-1,
k_{r}\geq2}}\frac{n!}{k_{1}!k_{2}!\ldots k_{r}!}=4\sum_{\substack{k_{1}%
+\ldots+k_{r}=n\\k_{i}\geq1,i=1,\ldots,r-1, k_{r}\geq2}}\binom{n}{k_{1}%
,\ldots,k_{r-1}}.
\]
\end{cor}
\begin{proof}
From Equation \ref{eq2.1}, for each choice $k_{1},\ldots,k_{r}$, with
condition $k_{1}+\ldots+k_{r}=n$, $k_{i}\geq1$, $i=1,\ldots,r-1$ and
$k_{r}\geq2$,

there are $\binom{n}{k_{1}}$ many ways to form $M_{1}$,

there are $\binom{n-k_{1}}{k_{2}}$ many ways to form $M_{2}$,

$\ldots$,

there are $\binom{n-k_{1}-\ldots-k_{r-1}}{k_{r}}$ many ways to form
$M_{r}$,

$a$ has two choices

$b$ has two choices.

Hence, the number is
\[
4\sum_{\substack{k_{1}+\ldots+k_{r}=n\\k_{i}%
\geq1,i=1,\ldots,r-1, k_{r}\geq2}}\binom{n}{k_{1}}%
\binom{n-k_{1}}{k_{2}}\ldots\binom{n-k_{1}-\ldots-k_{r-1}}{k_{r}}%
\]

\[
=4\sum_{\substack{k_{1}+\ldots+k_{r}=n\\k_{i}\geq1,i=1,\ldots,r-1,
k_{r}\geq2}}\frac{n!}{(k_{1})!(n-k_{1})!}\frac{(n-k_{1})!}{(k_{2}%
)!(n-k_{1}-k_{2})!}\ldots\frac{(n-k_{1}-\ldots-k_{r-1})!}{k_{r}!(n-k_{1}%
-\ldots-k_{r})!}%
\]

\[
=4\sum_{\substack{k_{1}+\ldots+k_{r}=n\\k_{i}\geq1,i=1,\ldots,r-1,
k_{r}\geq2}}\frac{n!}{k_{1}!k_{2}!\ldots k_{r}!}=4\sum_{\substack{k_{1}%
+\ldots+k_{r}=n\\k_{i}\geq1,i=1,\ldots,r-1, k_{r}\geq2}}\binom{n}{k_{1}%
,\ldots,k_{r-1}}.
\]
\end{proof}

\end{document}